\documentclass[]{article} 
\usepackage{mathtools}
\usepackage{float}
\usepackage[utf8]{inputenc}
\usepackage{verbatim}
\usepackage{smartdiagram}
\usepackage{todonotes}
\usepackage{caption}
\usepackage{amsthm}

\usepackage{graphicx}
\usepackage{multicol}
\usepackage{footmisc}
\usepackage{cite}
\usepackage{url}

\usepackage{enumerate}

\usepackage{amsmath, amssymb, amsfonts}

\usepackage{url}

\usepackage{amsmath, amssymb, amsfonts}

\usepackage[linesnumbered, boxed, lined, ruled,vlined]{algorithm2e}

\usetikzlibrary{arrows.meta,positioning}
\usetikzlibrary{shapes.geometric}

\newtheorem{definition}{Definition}[section]
\newtheorem{theorem}[definition]{Theorem}
\newtheorem{lemma}[definition]{Lemma}

\newcommand{\R}{\mathbb{R}}
\renewcommand{\O}{\mathcal{O}}

\author{T. Heller, S.O. Krumke}

\begin{document}          
	\title{Computing the egalitarian allocation with network flows}

	\author{T. Heller\footnote{\texttt{till.heller@itwm.fraunhofer.de}, corresponding author} \and S.O. Krumke}	
	
	\maketitle              
	
	\begin{abstract}
		In a combinatorial exchange setting, players place sell (resp. buy) bids on combinations of traded goods. Besides the question of finding an optimal selection of winning bids, the question of how to share the obtained profit is of high importance.
		The egalitarian allocation is a well-known solution concept of profit sharing games which tries to distribute profit among players in a most equal way while respecting individual contributions to the obtained profit. Given a set of winning bids, we construct a special network graph and show that every flow in said graph corresponds to a core payment. Furthermore, we show that the egalitarian allocation can be characterized as an almost equal maximum flow which is a maximum flow with the additional property that the difference of flow value on given edge sets is bounded by a constant. With this, we are able to compute the egalitarian allocation in polynomial time.  
\end{abstract}

\section{Introduction}
In a combinatorial exchange, players place sell (resp. buy) bids on combinations of traded goods. Combinatorial exchanges are often used to efficiently distribute various goods among a group of players, such as landing and take-off slots for airplanes (cf. \cite{rassenti1982combinatorial, gruyer2003auctioning}), collaborative planning (cf. \cite{hunsberger2000planning}) or in the area of transport logistics (cf. \cite{ledyard2002first, kalagnanam2001computational}). Besides the question of finding an optimal selection of winning bids, the question of how to share the obtained profit is of high importance.

\emph{Cooperative games} are a well-established approach for distributing profit among a group of cooperating players. A solution to such a cooperative game is given by a payment vector that stores the profit for each player. Since players wants to be treated in a fair way, one requires additional properties of the payment vector. One of these desired properties is equality between all players. Dutta and Ray introduced in \cite{dutta1989concept} the \emph{egalitarian allocation}, which is a payment vector that tries to split the available profit in a "most equal" way among the players. They gave both uniqueness and existence results. In order to compute the egalitarian allocation, they proposed an algorithm that requires to compute a maximally sized coalition that maximizes the average contribution in every iteration. In general, this needs an exponential number of evaluations of the characteristic function and, thus, the algorithm has an exponential running time in general.

In this paper, we show how to construct a network graph depending on the instance of the combinatorial exchange setting. We give an example that shows that an interpretation as a flow game on the whole instance yields in an empty core in general. Thus, we restrict ourselves to a solution of the combinatorial exchange. With this, we are able to construct a network graph in which every maximum flow induces a payment vector that is contained in the core of the corresponding minimum cost flow game. Furthermore, we exhibit an interesting connection between the egalitarian allocation being characterized by the almost equal flow with the smallest constant deviation. This provides a polynomial time algorithm for computing the egalitarian allocation. 

The rest of the paper is structured as follows. In the next section we give a brief introduction to core elements and the egalitarian allocation, as well as a definition of an \emph{almost equal flow}. In Section~\ref{sec: egal: flow games}, we provide the construction of the network graph as well as the characterization results. We conclude with a short outlook. 

\section{Preliminaries}
In this section we will briefly cover the preliminaries for the rest of the paper. 
\subsection{Game Theory}
A \emph{cooperative game} is given by a tuple~$(N,v)$, where $N$ denotes the set of players and $v: 2^N \mapsto \R$ denotes the \emph{characteristic function}, which maps each set of players (also called \emph{coalition}) to its value. A payment vector~$p\in\R^{|N|}$ denotes the payment to each of the players. It is said to fulfill \emph{individual rationality} (IR) if $p_i\geq 0$ holds for every player~$i$. The payment vector fulfills \emph{coalitional rationality} (CR) if $p_S = \sum_{i\in S} p_i \geq v(S)$ holds for every coalition~$S$. Furthermore, it fulfills \emph{efficiency} (EFF) if $p_N = v(N)$ holds and it fulfills \emph{budget balance} (BB) if $p_N \leq v(N)$ holds. The set of all payments that fulfill (IR), (CR) and (EFF) is called \emph{core} and is a well-established solution concept to cooperative games. The egalitarian allocation (EA) is special payment vector. It was shown that the EA always exists, in contrast to the core which can be empty. But, if a convex game is considered, the EA is always contained in the core (cf. \cite{dutta1989concept}). Furthermore, the EA is unique (cf. \cite{dutta1989concept}) and it is the lexicographical maximal vector in the core (cf. \cite{graef2021connection}). 

\subsection{Almost Equal Flows}\label{sec: aemfp}
Given a network graph~$G=(V,E)$ with a source~$s$ and a sink~$t$. Given \emph{homologous edge sets}~$R_i\subseteq E$ with a deviation value~$\delta_i$, the \emph{constant Almost-Equal-Maximum Flow Problem} (AEMFP) asks to find a maximum $s$-$t$-flow such that the difference of flow values on edges contained in the same homologous set is at most the deviation value. It was shown that the constant AEMFP can be solved by a strongly polynomial algorithm which depends on the parametric search technique of Megiddo (cf. \cite{haese2020algorithms}).

In the next section, we will define a flow game based on a combinatorial exchange solution. 

\section{Application to Flow Games}\label{sec: egal: flow games}

%


We consider a combinatorial energy exchange where profit is generated based on the flow of divisible items (e.g. energy) between a producer, i.e. a source node, and a consumer, i.e. a sink node. This can be understood as a minimum cost flow with negative costs which are given by the difference between buy and sell prices. In a first step we define the class of \emph{minimum cost flow games}, which is closely related to the \emph{flow games} (cf. \cite{kalai1982totally}). Let $G=(V,E)$ denote the \emph{exchange graph} with lower and upper capacity functions~$l,u: E\mapsto \R$ and a profit function~$c: E\mapsto \R_+$. Further, each node~$v$ has an associated balance~$b_v$. Now let $N$ denote the set of players. Each player controls some of the edges~$e\in E$. The payment to a player~$p$ is equal to the sum~$\sum_{e\in E_p} c_e f_e$, where $E_p$ denotes the set of edges controlled by player~$p$.

Note that one can show that the core is in general empty. To overcome this shortcoming, we apply the following transformation to our exchange graph to obtain the \emph{profit sharing graph} (see Figure~\ref{fig: transformation}). First, subdivide each of the edges~$e$ between time point nodes by introducing a node~$v$. Further, add a source node~$s$ to the graph and edges~$e=(s,v)$ to all introduced subdividing nodes~$v$ with upper capacity~$u_e = f^*(e)c(e)$. Then add a sink node~$t'$ to the graph and all contract nodes with their time point nodes. The contract nodes are still adjacent to their time point nodes and all contract nodes are adjacent to the sink~$t'$. The upper capacity of edges between a time point node and its contract node in profit sharing graph is given by the capacity of the edge between the same nodes in the exchange graph multiplied by the cost of the incident edge between two time point nodes.

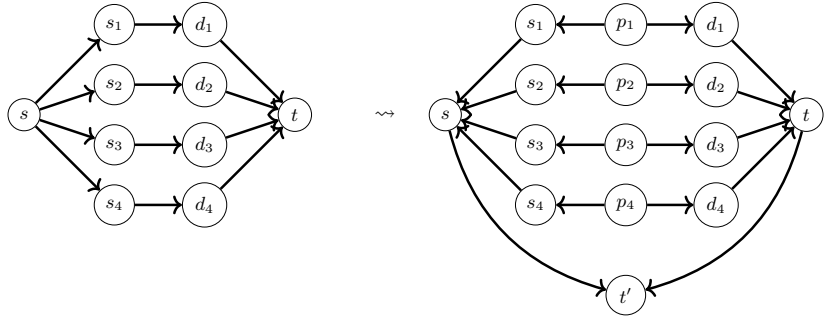
\begin{figure}
	\centering
	\scalebox{0.8}{
	\begin{tikzpicture}[every node/.style={fill=white,rectangle}, every edge/.style={draw=black,very thick}]
	\begin{scope}
	\node[draw, circle, font=\small] (s) at (0,0) {$s$};
	\node[draw, circle, font=\small] (t) at (4.5,0) {$t$};
	\node[draw, circle, font=\small] (s1) at (1.5,1.5) {$s_1$};
	\node[draw, circle, font=\small] (s2) at (1.5,0.5) {$s_2$};
	\node[draw, circle, font=\small] (s3) at (1.5,-0.5) {$s_3$};
	\node[draw, circle, font=\small] (s4) at (1.5,-1.5) {$s_4$};
	\node[draw, circle, font=\small] (d1) at (3,1.5) {$d_1$};
	\node[draw, circle, font=\small] (d2) at (3,0.5) {$d_2$};
	\node[draw, circle, font=\small] (d3) at (3,-0.5) {$d_3$};
	\node[draw, circle, font=\small] (d4) at (3,-1.5) {$d_4$};
	
	\path[->] (s) edge (s1);
	\path[->] (s) edge (s2);
	\path[->] (s) edge (s3);
	\path[->] (s) edge (s4);
	
	\path[->] (s1) edge (d1);
	\path[->] (s2) edge (d2);
	\path[->] (s3) edge (d3);
	\path[->] (s4) edge (d4);
	
	\path[->] (d1) edge (t);
	\path[->] (d2) edge (t);
	\path[->] (d3) edge (t);
	\path[->] (d4) edge (t);
	
	\end{scope}
	
	\node[draw=none] (e) at (6,0) {$\rightsquigarrow$};
	\begin{scope}[xshift=7cm]
	\node[draw, circle, font=\small] (s) at (0,0) {$s$};
	\node[draw, circle, font=\small] (d) at (6,0) {$t$};
	
	\node[draw, circle, font=\small] (t) at (3,-3) {$t'$};
	
	\node[draw, circle, font=\small] (s1) at (1.5,1.5) {$s_1$};
	\node[draw, circle, font=\small] (s2) at (1.5,0.5) {$s_2$};
	\node[draw, circle, font=\small] (s3) at (1.5,-0.5) {$s_3$};
	\node[draw, circle, font=\small] (s4) at (1.5,-1.5) {$s_4$};
	\node[draw, circle, font=\small] (d1) at (4.5,1.5) {$d_1$};
	\node[draw, circle, font=\small] (d2) at (4.5,0.5) {$d_2$};
	\node[draw, circle, font=\small] (d3) at (4.5,-0.5) {$d_3$};
	\node[draw, circle, font=\small] (d4) at (4.5,-1.5) {$d_4$};
	
	\node[draw, circle, font=\small] (p1) at (3,1.5) {$p_1$};
	\node[draw, circle, font=\small] (p2) at (3,0.5) {$p_2$};
	\node[draw, circle, font=\small] (p3) at (3,-0.5) {$p_3$};
	\node[draw, circle, font=\small] (p4) at (3,-1.5) {$p_4$};
	
	\path[<-] (s) edge (s1);
	\path[<-] (s) edge (s2);
	\path[<-] (s) edge (s3);
	\path[<-] (s) edge (s4);
	
	\path[->] (p1) edge (d1);
	\path[->] (p2) edge (d2);
	\path[->] (p3) edge (d3);
	\path[->] (p4) edge (d4);
	
	\path[->] (p1) edge (s1);
	\path[->] (p2) edge (s2);
	\path[->] (p3) edge (s3);
	\path[->] (p4) edge (s4);
	
	\path[->] (d1) edge (d);
	\path[->] (d2) edge (d);
	\path[->] (d3) edge (d);
	\path[->] (d4) edge (d);
	
	\path[->] (d) edge[bend left] (t);
	\path[->] (s) edge[bend right] (t);
	\end{scope}
	
	\end{tikzpicture}}
	\caption{Example of the transformation of an exchange graph (left) to a profit sharing graph (right).}\label{fig: transformation}
\end{figure}  

Before we continue, we first define the \emph{extended set of ingoing edges}~$\delta_{\text{ext}}^+(S)$ of a set~$S$ of contract nodes as the set of edges such that their end node is adjacent to a contract node in $S$. More formally,
\begin{align*}
\delta_{\text{ext}}^+(S) \coloneqq \{e\in E : \exists n\in S \text{ s.t. } \omega(e) \in N(n)\}.
\end{align*}

The payment to a contract~$b$ is now given by the sum of ingoing flow at the corresponding contract node~$n_b$, or, equivalently, as the flow~$f_{n_b, t}$ on the edge~$(n_b, s)$. 

\begin{theorem}[Core Elements]
	Any maximum flow in the profit sharing graph defines a payment that fulfills the properties (IR), (CR), and (BB).
\end{theorem}
\begin{proof}
	Clearly, since a flow can only attend non-negative values, the payment to each contract is also non-negative. Thus, (IR) is fulfilled.
	
	The value of a maximum flow on the profit sharing graph is by definition equal to the cost of the given optimal minimum cost flow in the exchange graph. Thus, it is not possible to distribute more profit than obtained by given solution -- (BB) is always fulfilled.
	
	In order to prove (CR), let $f$ denote a maximum flow in the profit sharing graph~$G'$. Then for any subset~$S$ of nodes we show that $\sum_{e\in\delta_{\text{ext}}^+(S)} f_e \geq v(S)$	holds true. First, by restricting ourselves to edges~$e$ such that all edges incident to the start node~$\alpha(e)$ are contained in $\delta_{\text{ext}}^+(S)$, we get
	\begin{align}
	\sum_{e\in\delta_{\text{ext}}^+(S)} f_e \geq \sum_{e\in\delta_{\text{ext}}^+(S) \text{ s.t. } \forall e' \text{ incident to } \alpha(e) : e'\in \delta_{\text{ext}}^+(S)} f_e.
	\end{align} 
	This is again equal to the sum~$\sum_{e\in\delta_{\text{ext}}+(S)}\sum_{e'\in E' : \omega(e') = \alpha(e)} u'_e$. By the definition of the profit sharing graph, we know that $u'_e = c_ef^*_e$ for the cost function~$c$ and the optimal flow~$f^*$ in the exchange graph~$G$.
	
	Further, the value~$v(S)$ of a coalition~$S$ is given by $\sum_{e\in P(s,d) : s,d\in S} c_ef^*_e$ where $P(s,d)$ denotes the edges on a path from a contract node~$s$ to a contract node~$d$ in the exchange graph~$G$. Finally, we showed that every maximum flow in the profit sharing graph induces a core element for a minimum cost flow game on a exchange graph.
\end{proof}


As we have seen above, we now know that by computing a maximum flow in the profit sharing graph, we obtain a core element for a minimum cost flow game defined on an exchange graph. Given that, we can use an iterative process to compute the EA. 

For this we use an iterative process that solves an equal maximum flow problem in every step. First, for a given equal maximum flow~$F^*$ we partition the set~$E_P$ of edges with the sink~$t$ as end node into two sets, the set~$F$ of fixed edges and the set~$H$ of non-fixed edges. We say an edge~$e$ is \emph{fixed} if there exists no other equal maximum flow~$f'$ with a flow value~$f'(e) > f^*(e)$. The algorithm now works as follows: Initialize the set of non-fixed edges by the set~$\delta^-(t)$, i.e. the set of ingoing edges to the sink~$t$, and the set of fixed edges as the empty set. Now compute an equal maximum flow on the profit sharing graph~$G_P$ where the one homologous edge set is given by the set of non-fixed edges (recall the definition of the AEMFP in Section~\ref{sec: aemfp}). 
%
\begin{lemma}
	Given a profit sharing graph with $n$ nodes and $m$ edges, the algorithm described above runs in $\O(m\cdot(T_{EMF}(n,m, 1)))$ time, where $T_{EMF}(n,m,k)$ denotes the time needed for computing an equal maximum flow on a graph with $n$ nodes, $m$ edges and $k$ homologous edge sets.
\end{lemma}
\begin{proof}
	In every iteration the algorithm fixes at least one edge since otherwise the computed equal maximum flow would not be maximal. Thus, at most~$m$ iterations are needed. In every iteration an equal maximum flow problem with one homologous edge set is solved. Overall, the running time is $\O(m\cdot(T_{EMF}(n,m, 1)))$.
\end{proof}
 Due to the following result, one can use algorithms for the AEMFP to compute the EA.
\begin{theorem}[Characterization as AEMF]
	The almost equal flow with the smallest deviation value corresponds to the EA.
\end{theorem}
\begin{proof}
	Let $a$ denote the EA and let $b$, $b\neq a$, be another vector in the core with a smaller difference between the largest and the smallest entry. We assume that both vectors are given in an increasing order, i.e. $a_i \leq a_{i+1}$. We want to show that the difference between the largest entry~$a_n$ and the smallest entry~$a_1$ of $a$ is smaller than the difference~$b_n - b_1$ for any other allocation vector~$b$. We know that the EA is the lexicographical maximal vector in the core (cf. \cite{}). 
	
	First, suppose the smallest entry of $b$ is strictly greater than the smallest entry of $a$ and $b_n - b_1 \leq a_n - a_1$ holds. This is a contradiction since $b$ is lexicographically greater than $a$ and $a$ is the lexicographical greatest vector. Suppose now $b_1 = a_1$, i.e. the smallest entries of both vectors are the same	and $b_n < a_n$, i.e. $b_n - b_1 < a_n - a_1$ holds. By construction each entry can be decreased by increasing exactly one other entry. This again is a contradiction to $a$ being the lexicographical maximal vector. 
	
	Now assume $b_1 < a_1$ and $b_n - b_1 < a_n - a_1$ holds. Let $\delta_n \coloneqq a_n -b_n$. By construction, each entry can be decreased by increasing exactly one other entry. Since $b$ exists, we know that there exists a vector with smaller largest entry. We also know that by a finite sequence of dual shifts one can construct $b$ from $a$. Thus, there exist shifts~$S_1,\dots, S_l$ that reduce the largest entry of $a$ by $\delta_n$. By construction, no real capacity constraints exist and the decreasing shifts are monotonous, i.e. each shift decreases the largest entry and no artificial shifts of two other entries are necessary. We distinguish different cases.
	
	First, assume that the entry $a_n - \delta_n$ is still the largest entry in the vector~$a'$ constructed by the sequence of shifts~$S_1,\dots,S_l$. Since all shifts are monotonous, no other entry was decreased by these shifts. Thus, the vector~$a'$ is lexicographically larger than $a$ --- a contradiction. Now assume that the entry $a_n-\delta_n$ is not the largest entry in the constructed vector~$a'$. Then either the largest element in $a'$ is equal to $a_n$ or not. In the first case we can also apply a sequence of shifts that reduces the entry~$a'_n$ by $\delta_n$ and thus end up in the latter case where the largest entry~$a'_n$ is strictly smaller than $a_n$. Since all made shifts are monotonous, i.e. all other entries are only increasing, we end up with a vector~$a'$ that is lexicographically greater than $a$ --- again a contradiction.  
	%
\end{proof}

\section{Outlook}
In this paper we showed how to compute the EA in the context of flow games. Furthermore, we showed that the EA allows a characterization as an almost equal maximum flow with the smallest deviation value. With this, one can obtain strongly polynomial algorithms for computing the EA.

\bibliographystyle{plain}
\bibliography{references.bib}

\end{document}